\documentclass[journal]{IEEEtran}

\ifCLASSINFOpdf
  \usepackage[pdftex]{graphicx}

\else

  \usepackage[dvips]{graphicx}

\fi
\usepackage{subfig}
\usepackage{xcolor}
\usepackage{multirow}
\usepackage{booktabs}
\usepackage{pifont}

\usepackage{cite}
\usepackage{amsmath,amssymb,amsfonts}
\usepackage{amsthm}
\usepackage{algorithmic}
\usepackage{graphicx}
\usepackage{textcomp}
\usepackage{xcolor}
\usepackage{booktabs}
\usepackage{multirow}

\usepackage{amsthm}

\usepackage{caption}
\newtheorem{theorem}{Theorem}

\newtheorem{proposition}{Proposition}
\newtheorem{corollary}{Corollary}

\theoremstyle{remark}
\newtheorem{remark}{Remark}

\def\BibTeX{{\rm B\kern-.05em{\sc i\kern-.025em b}\kern-.08em
    T\kern-.1667em\lower.7ex\hbox{E}\kern-.125emX}}

\begin{document}

\title{Radiomap Blind Prediction under Incomplete Observation: Error Characterization and Correctable Propagation-Prior Learning}

\author{Xiaojie Li,~\IEEEmembership{Student Member,~IEEE,}
Yu Han,~\IEEEmembership{Member,~IEEE,}
Han Fang,
Shangqing Liu,
\\
Guangxu Zhu,~\IEEEmembership{Member,~IEEE,}
Shi Jin,~\IEEEmembership{Fellow,~IEEE,}
and Chao-Kai Wen,~\IEEEmembership{Fellow,~IEEE}
\thanks{
A preliminary version of this work, entitled
``Rethinking Radiomap Blind Prediction with Limited Environment and
Configuration Representations,'' has been accepted for presentation at
the 2026 IEEE Global Communications Conference (GLOBECOM).

Xiaojie Li, Yu Han, and Shi Jin are with the School of Information Science and Engineering, Southeast University, Nanjing 210096, China, and also with National Mobile Communications Research Laboratory, Southeast University, Nanjing 210096, China (e-mail: \{xiaojieli, hanyu, and jinshi\}@seu.edu.cn).

Han Fang is with the School of Cyber Science and Engineering, Southeast University, Nanjing 210096, China (e-mail: h\_fang@seu.edu.cn).

Shangqing Liu is with the State Key Laboratory of Novel Software Technology, Nanjing University, Nanjing 210023, China (e-mail: shangqingliu@nju.edu.cn).

Guangxu Zhu is with the Shenzhen Research Institute of Big Data, The Chinese University of Hong Kong-Shenzhen, Guangdong 518172, China (e-mail: gxzhu@sribd.cn).

Chao-Kai Wen is with the Institute of Communications Engineering,
National Sun Yat-sen University, Kaohsiung 80424, Taiwan (e-mail:
chaokai.wen@mail.nsysu.edu.tw).}

}

\markboth{IEEE Transactions on Wireless Communications}%
{Shell \MakeLowercase{\textit{et al.}}: Bare Demo of IEEEtran.cls for IEEE Journals}

\maketitle

\begin{abstract}
Radiomap blind prediction aims to infer radiomaps from observable representations of the propagation environment and base station configuration without field measurements. In practice, the observable representations are inherently incomplete. Thus, the target radiomap is not fully determined by the inputs when generalizing to unseen configurations or environments. Under incomplete observation, we establish a population-level theory of deterministic radiomap blind prediction that identifies the conditional mean as its optimal target and separates prediction error into reducible predictor approximation and irreducible uncertainty caused by missing physical information. The framework further characterizes the train-test risk gap and the uncertainty reduction enabled by observation enrichment. Building on it, we reveal the dual role of propagation priors: they provide physically grounded guidance, yet their implementable forms may bias the attainable predictor. This motivates RadioDecomp, which treats a prior-guided predictor as a correctable base and learns its remaining predictable discrepancy through residual refinement. To evaluate RadioDecomp across distinct propagation-prior designs, we instantiate it with a feature-guided monolithic base and a LoS-Shadow structured base, yielding RadioFR and RadioLSR, respectively. Across random, cross-configuration, and cross-environment settings, experiments confirm the benefit of propagation-related representations and show that both instantiations improve upon their respective bases. Further controlled studies on base capacity, training-support coverage, and observation coarsening corroborate the proposed analysis.

\end{abstract}


\section{Introduction}
\IEEEPARstart{F}{uture} 6G networks are expected to rely on environment-aware radio intelligence, creating strong demand for radiomaps that characterize the spatial distribution of radio attributes such as received power and angle of arrival \cite{zeng2024tutorial}. Such radiomaps are valuable for wireless network planning\cite{10877906}, optimization\cite{10299600}, predictive radio management\cite{11300801,11418311}, and covert communications\cite{11511762,11075592}. In practice, obtaining radiomaps through exhaustive measurements is prohibitively costly, motivating extensive research on radiomap reconstruction~\cite{ref3,ref4,ref5}. However, many applications require reliable radiomaps for areas or network configurations without available radiomap measurements, such as coverage prediction for base station (BS) deployment, reconfiguration, or decommissioning. This gives rise to \emph{radiomap blind prediction}, which predicts radiomaps solely from observable representations of the propagation environment and BS configuration~\cite{bi2019engineering,li2026u6gxlmimoradiomapprediction}.

Traditional radiomap blind prediction mainly relied on statistical propagation models~\cite{ref6,ref7} or ray tracing~\cite{ref8}. Statistical models infer radiomaps from empirical propagation laws calibrated for specific scenario classes or from stochastic models of network geometry, enabling efficient estimation but limited site-specific characterization. Ray tracing instead simulates propagation paths within a locally modeled scene, offering finer radiomap characterization but requiring detailed environment and BS modeling and substantial computation. Enabled by advances in deep learning, learning-based methods have emerged as a critical alternative approach. Such methods directly learn the mapping from observable environment and BS-configuration representations to their corresponding radiomaps. Early studies mainly investigated radiomap blind prediction across varying environments under fixed or simplified BS configurations~\cite{radiounet,rmegan,radiodiff}. Subsequent studies extended this setting to jointly varying environments and configurations, explicitly accounting for variations in carrier frequency, antenna patterns, array configurations, and beamforming~\cite{ref13,ref14,ref15,ref16,ref17,li2026u6gxlmimoradiomapprediction}. However, only partial environment and configuration representations are available in practical blind prediction.

Such incomplete representations omit propagation-relevant factors, such as material properties, fine-grained geometry, small scatterers, indoor structures, and hardware imperfections. Consequently, the same observable representation may be compatible with multiple complete physical states that produce different radiomaps. Nevertheless, radiomap blind prediction is commonly formulated using deterministic predictors that return a single estimate for each observable input. This contrast raises fundamental questions about the statistical target approximated by such a predictor and about how its prediction error and generalization should be understood under incomplete observation.

Understanding these issues requires distinguishing radiomap prediction under
incomplete observation from radiomap generation under a complete physical state. Theoretically, the observable and hidden factors together define a complete physical state. Under such state, the radiomap is deterministically generated through electromagnetic propagation ultimately governed by Maxwell's equations. Thus, the physical regularities make propagation knowledge a natural and principled source of structural guidance for radiomap blind prediction, particularly under finite training data and unseen test conditions.

Existing methods incorporate such propagation knowledge at the feature, model, and loss levels. Feature-level methods transform propagation knowledge into model inputs~\cite{ref20,ref21,ref22,ref23,ref24}, model-level methods encode propagation processes into network structures~\cite{ref25,ref26,ref27,ref28}, and loss-level methods constrain model training using propagation rules~\cite{ref29,ref30,ref31}. These mechanisms utilise useful propagation regularities and reduce the difficulty of learning from finite samples\cite{li2026crossdomainradiomappredictionmultiscatterer}. However, under incomplete observation, the complete propagation process cannot be fully characterized from the observable inputs, while implementable propagation priors typically capture selected or simplified physical regularities. Consequently, such priors may be inconsistent with the statistically optimal prediction relation, thereby introducing prior-induced bias. Existing studies have primarily emphasized their empirical benefits, leaving this benefit-bias tradeoff under incomplete observation insufficiently understood.

To address these issues, we develop a population-level framework for radiomap blind prediction under incomplete observation. We formulate radiomap generation using observable and hidden physical factors and identify the conditional mean given the observable representation as the optimal deterministic predictor under squared loss, providing a statistical basis for understanding prediction error and train-test generalization. We further examine propagation-prior injection relative to this conditional mean target and show that, although propagation priors can facilitate finite-sample learning, their partial or simplified physical assumptions may introduce prior-induced bias. This observation motivates \emph{RadioDecomp}, which retains a propagation-prior-guided predictor as its base and introduces a residual predictor to correct its deviation from the conditional target. The main contributions of this paper are summarized as follows:

\begin{itemize}
\item \textbf{Incomplete-Observation Error Analysis Framework.}
We establish a population-level framework to distinguish
radiomap blind prediction errors addressable by predictor improvement from those caused by incomplete observation. Under squared loss, we identify the conditional-mean radiomap as the optimal prediction target and decompose the expected error of any deterministic predictor into its approximation error relative to this target and irreducible uncertainty. The framework further characterizes how these two terms contribute to the train-test risk gap and proves that additional observations strictly reduce the uncertainty whenever they change the conditional-mean radiomap with positive probability.

    \item \textbf{Propagation-Prior Benefit-Bias Analysis.}
    Building on this framework, a unified analysis of feature-, model-, and loss-level prior injection characterizes information-compression, structural-approximation, and loss-induced biases. Together, these results show how propagation priors can facilitate finite-sample learning while their implementable forms may bias the attainable predictor.

\item \textbf{RadioDecomp: Correctable Propagation-Prior Learning.}
We propose RadioDecomp, which retains a prior-guided radiomap predictor as a base and learns its remaining predictable discrepancy through a residual branch. Two instantiations examine complementary correction needs: RadioFR addresses prediction errors that remain under information-preserving feature augmentation, while RadioLSR addresses discrepancies under a more restrictive LoS-Shadow base with region-specific inputs. Experiments under random, cross-configuration, cross-environment, and limited-data settings demonstrate the benefits of prior guidance and residual correction, while studies of training-support coverage and observation coarsening support the theoretical analysis.
\end{itemize}

\section{Related Work}
\label{sec:related_work}

\subsection{Learning-Based Radiomap Prediction}
\label{subsec:learning_based_prediction}
Early blind-prediction methods mainly modeled environmental variations under fixed or simplified transmitter configurations (isotropic antenna pattern). RadioUNet~\cite{radiounet} established an image-to-image prediction paradigm using building maps and transmitter locations, while RME-GAN~\cite{rmegan} employed a conditional generative adversarial network for radiomap prediction. RadioDiff~\cite{radiodiff} subsequently introduced a conditional diffusion model for measurement-free prediction.

More recent studies have extended radiomap learning toward configuration awareness. Directional antenna patterns were incorporated through spatial gain projections~\cite{ref13}, while UniRM~\cite{ref14} studied adaptation across frequency bands and receiver heights with sparse target measurements. BeamCKM~\cite{ref15} and BeamCKMDiff~\cite{ref16} modeled discrete and continuous beamforming configurations, respectively. Configuration-dependent coverage was also investigated through disentangled representations using real measurements~\cite{ref17} and explicit beam-map representations for U6G XL-MIMO systems~\cite{li2026u6gxlmimoradiomapprediction}. These studies demonstrate a progression from environment-aware to joint environment- and configuration-aware prediction, but the statistical target of deterministic blind prediction under incomplete observation remains insufficiently understood.

\subsection{Propagation Priors for Radiomap Prediction}
\label{subsec:propagation_prior_related_work}

Propagation knowledge has been incorporated at the feature, model, and loss levels. Feature-level methods provide models with propagation-related representations derived from geometry, materials, attenuation, obstruction, or wave equations~\cite{ref20,ref21,ref22,ref23,ref24}. Such features can expose useful propagation patterns to a radiomap predictor but may be inaccurate because they depend on simplified models and incomplete observation.

Model-level methods embed propagation assumptions into predictor structures. Existing examples include ray-tracing-informed large-model modules~\cite{ref25}, geometry-aware decomposition of propagation components~\cite{ref26}, reciprocity-constrained bidirectional Gaussian splatting~\cite{ref27}, physics-guided diffusion architectures~\cite{ref28}, and disentangled modeling of coverage factors~\cite{ref17}. These structures can improve data efficiency, but structural mismatch may restrict the hypothesis class and limit its approximation of the optimal predictor.

Loss-level methods incorporate propagation knowledge through physical constraints or propagation-aware error metrics during training. Representative approaches use volume-integral equations~\cite{ref29}, LoS power constraints and NLoS propagation-delay lower bounds~\cite{ref30}, Helmholtz-equation and boundary residuals~\cite{ref31}, or logarithmic-power-domain prediction losses~\cite{ref27}. Approximate constraints and their weights may, however, shift the population optimum of the original prediction task.

Totally, existing studies primarily demonstrate the empirical benefits of task-specific propagation priors. A unified understanding of their finite-sample benefits and prior-induced biases under incomplete observation, together with a general mechanism for correcting such biases, is still lacking.

\section{Error Characterization under Incomplete Observation}
\label{sec:error_anatomy}

\subsection{From Physical Generation to Blind Prediction}
\label{subsec:full_to_blind}

Let \(E\) denote the complete physical environment and \(C\) the complete BS configuration. The former includes propagation-relevant scene information, such as geometry and material properties, while the latter specifies the BS location and transmission setup, including the antenna array, carrier frequency, and beamforming configuration. A radiomap is a discretized spatial representation of a radio attribute, such as received power, angle of arrival, or angle of departure. In the 2D setting considered here, a radiomap realization is represented by \(\mathbf y\in\mathbb R^{H\times W}\). The same formulation can be extended to voxelized 3D radiomaps~\cite{10963917,11455177}. When \(E\) and \(C\) are fully specified, radiomap generation can be
expressed as
\begin{equation}
    \mathbf y=F(E,C),
    \label{eq:full_generation}
\end{equation}
where \(F\) denotes the deterministic propagation mechanism. Therefore, full-information radiomap generation is a deterministic forward problem, which can in principle be numerically solved using full-wave simulation\cite{10173521} or high-fidelity ray tracing.

Radiomap blind prediction differs because the predictor receives only incomplete observation of the environment and BS configuration. The observable environment may omit fine-scale geometry, electromagnetic material properties, indoor structures, or small scatterers. Likewise, recorded BS configurations may contain only coarse parameters, such as the beam direction, 3-dB beamwidth, and nominal transmit power, while the exact beam pattern, hardware-induced power offsets, and other site-specific calibration parameters remain unrecorded.

Let \(\mathbf Z=(E_{\mathrm{obs}},C_{\mathrm{obs}})\) denote the observable representation, taking values in the observable-input space \(\mathcal Z\). Elements of \(\mathcal Z\) may combine spatial environment representations with discrete and continuous BS-configuration variables. Let \(\mathbf H\) collect the remaining hidden propagation factors. The complete physical state can then be partitioned into \((\mathbf Z,\mathbf H)\). Accordingly,
\begin{equation}
    \mathbf Y=F(\mathbf Z,\mathbf H).
    \label{eq:zyh_generation}
\end{equation}

For a given observable representation \(\mathbf z\), multiple hidden states may be compatible with it and produce different radiomaps. For example, scenes with the same observable geometry may differ in material properties or unobserved scatterers. Accordingly,\footnote{Throughout this paper, bold uppercase symbols denote random variables, and their lowercase counterparts denote realizations. We use \(p(\cdot)\) generically for probability laws, including discrete, continuous, and mixed cases.}
\begin{equation}
\mathbf Y\mid \left(\mathbf Z=\mathbf z\right)=F(\mathbf z,\mathbf H),\qquad\mathbf H\sim p(\cdot\mid\mathbf z).
\label{eq:conditional_generation}
\end{equation}

Although \(\mathbf H\) is unobserved, a deterministic blind predictor must return a single radiomap for each observable input \(\mathbf z\). This single-valued mapping is learned from paired observations of \(\mathbf Z\) and \(\mathbf Y\). Under squared loss, its population-optimal output is
\begin{equation}
\begin{aligned}
m^\star(\mathbf z)&:=\arg\min_{\widehat{\mathbf y}}\mathbb E\left[\left\| \mathbf Y-\widehat{\mathbf y}\right\|_{\mathrm F}^{2}\mid\mathbf Z=\mathbf z\right] \\
&=\mathbb E\left[\mathbf Y\mid\mathbf Z=\mathbf z\right]=\mathbb E_{\mathbf H\sim p(\cdot\mid\mathbf z)}\left[F(\mathbf z,\mathbf H)\right].
\end{aligned}
\label{eq:population_conditional_mean}
\end{equation}

\begin{remark}[Learning target and error-analysis perspective]
\label{rem:population_target}
For a given data population, deterministic radiomap blind prediction aims to learn a mapping \(f:\mathcal Z\rightarrow\mathbb R^{H\times W}\) from the available observation to a single radiomap estimate. \textbf{Under squared loss, the ideal predictor is the conditional-mean mapping \(m^\star\), and the learning objective is to approximate this mapping as closely as possible.} Even if \(f(\mathbf Z)=m^\star(\mathbf Z)\) almost surely, its prediction may still differ from the realized radiomap \(\mathbf Y\), because hidden propagation factors remain unresolved by the observation. The following analysis therefore uses \(m^\star\) as the reference predictor to distinguish the error due to imperfect approximation by \(f\) from the irreducible uncertainty under incomplete observation, and examines how these two terms vary between training and test domains.
\end{remark}

\subsection{Domain-Wise Risk Decomposition and Train-Test Gap}
\label{subsec:domain_wise_risk}
A radiomap predictor is often trained on available environments and BS configurations and then applied to unseen ones. Its prediction error may change because both the observable environment/configuration distribution and the unrecorded propagation factors differ between training and test deployment. Let \(d\in\{\mathrm{tr},\mathrm{te}\}\) index the training and test domains, respectively. Each domain is characterized by a joint distribution \(p_d(\mathbf z,\mathbf h)\) over the observable input and the unobserved propagation factors. Accordingly, the population-optimal deterministic predictor under squared loss is domain-dependent:
\begin{equation}
m_d^\star(\mathbf z):=\mathbb E_d
\left[\mathbf Y\mid\mathbf Z=\mathbf z
\right]=\mathbb E_{\mathbf H\sim p_d(\cdot\mid\mathbf z)}\left[F(\mathbf z,\mathbf H)\right].
\label{eq:domain_conditional_mean}
\end{equation}
The corresponding conditional uncertainty is defined as
\begin{equation}
v_d(\mathbf z):=\mathbb E_d\left[\left\|\mathbf Y-m_d^\star(\mathbf z)\right\|_{\mathrm F}^{2}\mid\mathbf Z=\mathbf z\right].
\label{eq:domain_conditional_uncertainty}
\end{equation}
The conditional target \(m_d^\star\) describes the deterministic prediction target \(m^\star\) in domain \(d\), whereas \(v_d\) measures the remaining radiomap variation induced by the unobserved factors. These functions may differ across domains when the conditional hidden-factor distributions differ. Independently, different observable-input distributions determine how the corresponding pointwise quantities are aggregated within each domain. The following theorem shows that these two quantities yield an exact decomposition of the domain risk.

\begin{theorem}[Domain-wise risk decomposition]
\label{thm:domain_risk_decomposition}
For any deterministic predictor \(f\) with finite second moment and any
\(d\in\{\mathrm{tr},\mathrm{te}\}\), the population risk satisfies
\begin{equation}
\begin{aligned}
R_d(f)
&:=
\mathbb E_d\left[\left\|\mathbf Y-f(\mathbf Z)\right\|_{\mathrm F}^{2}\right] \\&=\underbrace{\mathbb E_d\left[\left\|f(\mathbf Z)-m_d^\star(\mathbf Z)\right\|_{\mathrm F}^{2}\right]}_{\mathcal A_d(f)}+\underbrace{\mathbb E_d\left[v_d(\mathbf Z)\right]}_{\mathcal U_d},
\end{aligned}
\label{eq:domain_risk_decomposition}
\end{equation}
where \(\mathcal A_d(f)\) is the domain-specific approximation error and \(\mathcal U_d\) is the irreducible uncertainty induced by incomplete observation in domain \(d\).
\end{theorem}

\begin{proof}
Adding and subtracting \(m_d^\star(\mathbf Z)\) gives
\begin{equation}
\begin{aligned}
R_d(f)={}&\mathbb E_d\left[\left\|f(\mathbf Z)-m_d^\star(\mathbf Z)\right\|_{\mathrm F}^{2}\right]+\mathbb E_d\left[\left\|\mathbf Y-m_d^\star(\mathbf Z)\right\|_{\mathrm F}^{2}\right]\\
&+2\mathbb E_d\left[\left\langle f(\mathbf Z)-m_d^\star(\mathbf Z), m_d^\star(\mathbf Z)-\mathbf Y \right\rangle_{\mathrm F}\right].
\end{aligned}
\end{equation}
The cross term is zero because
$\mathbb E_d\left[\mathbf Y-m_d^\star(\mathbf Z)\mid\mathbf Z\right]=0$. The second term equals \(\mathbb E_d[v_d(\mathbf Z)]\) by the law of total expectation, which proves \eqref{eq:domain_risk_decomposition}.
\end{proof}

Theorem~\ref{thm:domain_risk_decomposition} separates two sources of radiomap prediction error with different physical implications. The approximation term \(\mathcal A_d(f)\) measures how effectively the predictor exploits the recorded environment and BS configuration, whereas \(\mathcal U_d\) captures the remaining radiomap variability associated with hidden propagation factors. For a fixed observation and domain, \(R_d(f)\geq\mathcal U_d\), with equality when the predictor recovers \(m_d^\star\) almost surely. Hence, increasing model capacity or improving training can bring the prediction closer to the conditional-mean radiomap, but cannot resolve propagation differences that the input does not distinguish. For example, a geometry-only predictor may learn typical attenuation and blockage patterns while retaining uncertainty associated with unrecorded material properties. The decomposition therefore distinguishes the benefit attainable through predictor improvement from that requiring richer physical observations. Since both terms depend on the deployment domain, their changes jointly determine the train-test risk gap.

\begin{corollary}[Train-test risk gap]
\label{cor:risk_gap}
For any fixed deterministic predictor \(f\),
\begin{equation}
\begin{aligned}
R_{\mathrm{te}}(f)-R_{\mathrm{tr}}(f)={}&\underbrace{\mathcal A_{\mathrm{te}}(f)-\mathcal A_{\mathrm{tr}}(f)}_{\text{domain-specific approximation gap}}
\\
&+\underbrace{\mathcal U_{\mathrm{te}}-\mathcal U_{\mathrm{tr}}}_{\text{conditional-uncertainty gap}}.
\end{aligned}
\label{eq:train_test_risk_gap}
\end{equation}
\end{corollary}

Corollary~\ref{cor:risk_gap} explains why radiomap prediction performance can change substantially when a model is transferred to new environments or BS configurations. The approximation gap reflects how well the learned mapping captures propagation behavior in the test domain. It can increase when training data provide limited coverage of the relevant layouts or configurations, or when different hidden propagation conditions change the conditional-mean radiomap. The uncertainty gap captures a complementary effect: the same observable representation may leave different amounts of propagation variability unresolved across domains. For example, a building-geometry map may provide a more complete description of propagation in a domain with relatively uniform materials than in one with substantial unrecorded material variation. The train-test gap therefore reflects both the transferability of the learned propagation relation and the adequacy of the available physical description in each domain.

In practice, learning is primarily supported by the available training radiomaps. Broader coverage of training data (relevant environments and BS configurations) helps the predictor learn propagation relations in the training domain, while effective predictor design determines how well these relations are extracted from finite samples. The uncertainty term points to a separate question: can additional environment or configuration information make the target radiomap more predictable? This motivates the following analysis of observation enrichment.

\subsection{Observation Enrichment and Conditional Uncertainty}
\label{subsec:observation_enrichment}

For radiomap blind prediction, observation enrichment means providing additional propagation-relevant information, such as adding material properties to a building-geometry map or supplying a detailed antenna pattern alongside a coarse beam descriptor. Such information may distinguish propagation conditions that were indistinguishable under the original input, thereby lowering the minimum achievable radiomap prediction error.

Let \(\tilde{\mathbf Z}\) denote an enriched observation from which the original observation can be recovered:
\begin{equation}
\mathbf Z=\psi\left(\tilde{\mathbf Z}
\right).
\label{eq:enrichment_recovery}
\end{equation}
Thus, \(\tilde{\mathbf Z}\) retains all information contained in \(\mathbf Z\) and may additionally reveal propagation-relevant factors.

For each domain \(d\in\{\mathrm{tr},\mathrm{te}\}\), define the conditional mean and conditional uncertainty under the enriched observation as
\begin{equation}
\tilde m_d^\star(\tilde{\mathbf z}):=\mathbb E_d\left[\mathbf Y\mid\tilde{\mathbf Z}=\tilde{\mathbf z}\right]
\label{eq:enriched_conditional_mean}
\end{equation}
and
\begin{equation}
\tilde v_d(\tilde{\mathbf z}):=\mathbb E_d\left[\left\|\mathbf Y-\tilde m_d^\star(\tilde{\mathbf z})\right\|_{\mathrm F}^{2}\mid\tilde{\mathbf Z}=\tilde{\mathbf z}\right].
\label{eq:enriched_conditional_uncertainty}
\end{equation}
Let
\begin{equation}
\tilde{\mathcal U}_d:=\mathbb E_d\left[\tilde v_d(\tilde{\mathbf Z})\right].
\label{eq:enriched_domain_uncertainty}
\end{equation}

\begin{proposition}[Monotonicity under observation enrichment]
\label{prop:observation_enrichment}
Suppose the involved random variables have finite second moments and \(\mathbf Z=\psi(\tilde{\mathbf Z})\). Then, for every domain \(d\in\{\mathrm{tr},\mathrm{te}\}\),
\begin{equation}
\begin{aligned}
v_d(\mathbf Z)
&-\mathbb E_d\left[\tilde v_d(\tilde{\mathbf Z})\mid\mathbf Z\right]\\
&=\mathbb E_d\left[\left\|\tilde m_d^\star(\tilde{\mathbf Z})-m_d^\star(\mathbf Z)\right\|_{\mathrm F}^{2}\mid\mathbf Z\right]\ge 0
\end{aligned}
\label{eq:conditional_enrichment_gain}
\end{equation}
almost surely. Consequently,
\begin{equation}
\begin{aligned}
\mathcal U_d-\tilde{\mathcal U}_d
&=\mathbb E_d\left[\left\|\tilde m_d^\star(\tilde{\mathbf Z})-m_d^\star(\mathbf Z)\right\|_{\mathrm F}^{2}\right]\ge 0.
\end{aligned}
\label{eq:observation_enrichment_gain}
\end{equation}
\end{proposition}

\begin{proof}
Because \(\mathbf Z=\psi(\tilde{\mathbf Z})\), \(m_d^\star(\mathbf Z)\) is also determined by \(\tilde{\mathbf Z}\). Expanding \[\mathbf Y-m_d^\star(\mathbf Z)=\mathbf Y-\tilde m_d^\star(\tilde{\mathbf Z})+\tilde m_d^\star(\tilde{\mathbf Z})-m_d^\star(\mathbf Z)\]
and conditioning first on \(\tilde{\mathbf Z}\) eliminates the cross term because \[\mathbb E_d\left[\mathbf Y-\tilde m_d^\star(\tilde{\mathbf Z})\mid \tilde{\mathbf Z}\right]=0.\] Conditioning subsequently on \(\mathbf Z\) proves \eqref{eq:conditional_enrichment_gain}. Taking expectations proves
\eqref{eq:observation_enrichment_gain}.
\end{proof}

Proposition~\ref{prop:observation_enrichment} gives a precise criterion for the value of additional physical observations: the uncertainty reduction equals the expected squared change in the conditional-mean radiomap. Thus, additional information is useful for lowering the prediction limit when it distinguishes propagation conditions that were previously averaged together. For example, recording material properties can distinguish attenuation behavior among scenes with the same geometry, while a detailed antenna pattern can distinguish coverage variations hidden by a coarse beam descriptor. The relevant consideration is therefore the additional propagation information revealed by the representation.

Together, these results distinguish the roles of training-data acquisition, observation design, and predictor learning. More training radiomaps and broader coverage of environment and configuration variations support learning the predictable propagation relation; richer observations can make previously unresolved radiomap variations predictable. With the available observations fixed, the remaining task is to exploit their propagation information effectively from finite training data. Propagation priors offer a natural source of guidance by making selected physical relations explicit, motivating the following analysis of their benefits and potential biases.

\section{Propagation-Prior Guidance and Bias}
\label{sec:prior_under_incomplete_observation}

\subsection{Propagation Guidance for Predictor Approximation}
\label{subsec:propagation_prior_definition}

With the available environment and BS information fixed,
radiomap predictor design aims to reduce the approximation term \(\mathcal A_d(f)\). Given training data, minimizing the training-domain population risk is equivalent to reducing
\begin{equation}
\mathcal A_{\mathrm{tr}}(f)=\mathbb E_{\mathrm{tr}}\left[\left\|f(\mathbf Z)-m_{\mathrm{tr}}^\star(\mathbf Z)\right\|_{\mathrm F}^{2}\right],
\label{eq:training_approximation_objective}
\end{equation}
that is, learning a predictor that approximates the training-domain conditional target \(m_{\mathrm{tr}}^\star\). In practice, empirical training error can often be substantially reduced, for example, by increasing model capacity or improving optimization. However, a small empirical training error does not necessarily imply that \(f\) closely approximates \(m_{\mathrm{tr}}^\star\) at the population level, much less that it approximates \(m_{\mathrm{te}}^\star\) well in the test domain. To identify a possible source of test-domain guidance, recall that the conditional targets in both domains can be written as
\begin{equation}
m_{d}^\star(\mathbf z)=\mathbb E_{\mathbf H\sim p_{d}(\cdot\mid\mathbf z)}\left[F(\mathbf z,\mathbf H)
\right],\qquad d\in\{\mathrm{tr},\mathrm{te}\}.\label{eq:domain_target_propagation_view}
\end{equation}
Although the conditional distributions of the hidden factors may differ between the two domains, both targets are governed by the same physical propagation mechanism \(F\). Partial knowledge of \(F\) can therefore provide domain-relevant guidance for learning relations that may transfer from the training samples to unseen test inputs. This makes propagation knowledge a natural inductive prior for improving test-domain approximation.

In this paper, \(\Phi\) denotes a \emph{propagation prior}, namely, partial structural knowledge of \(F\), such as distance-dependent attenuation, antenna radiation, LoS propagation, blockage, reflection, diffraction, or scattering. It may be incorporated through input features, model structures, training losses, or their combinations. The following analysis examines how these different prior-injection mechanisms affect the training-domain learning problem, including their population solutions and possible induced biases.

\subsection{Prior-Injection Forms and Induced Biases}
\label{subsec:prior_forms}

\subsubsection{Feature-level prior}
For radiomap blind prediction, feature-level injection converts available environment and BS information into propagation-related inputs, such as geometry-derived blockage indicators, distance-based attenuation maps, or beam-gain maps computed from recorded antenna and beamforming parameters. For a fixed prior \(\Phi\), this
construction is represented as

\begin{equation}
\mathbf U=T_{\Phi}(\mathbf Z),
\label{eq:feature_prior_representation}
\end{equation}
where \(T_\Phi\) makes selected propagation relations explicit to the radiomap predictor. The resulting features may replace the original observations or be supplied alongside them, yielding \emph{prior-only} and \emph{prior-augmented} inputs, respectively.

\textit{1) Prior-only input:}
Consider a predictor \(g(\mathbf U)\) that uses only \(\mathbf U=T_{\Phi}(\mathbf Z)\). Under training-domain squared loss, its population conditional target is
\begin{equation}
\begin{aligned}
m_{\Phi,\mathrm{tr}}^\star(\mathbf u)
&:=\mathbb E_{\mathrm{tr}}\left[\mathbf Y\mid \mathbf U=\mathbf u\right]=\mathbb E_{\mathrm{tr}}\left[\mathbb E_{\mathrm{tr}}\left[\mathbf Y\mid \mathbf Z\right]\middle|\mathbf U=\mathbf u\right]
\\&=\mathbb E_{\mathrm{tr}}\left[m_{\mathrm{tr}}^\star(\mathbf Z)\mid T_{\Phi}(\mathbf Z)=\mathbf u\right].
\end{aligned}
\label{eq:feature_only_conditional_target}
\end{equation}
The second equality follows because \(\mathbf U=T_{\Phi}(\mathbf Z)\) is measurable with respect to \(\mathbf Z\).

To compare this target with \(m_{\mathrm{tr}}^\star(\mathbf z)\) on the original observable space, define its pullback to the \(\mathbf Z\)-space as
\begin{equation}
\bar m_{\Phi,\mathrm{tr}}^\star(\mathbf z)
:=m_{\Phi,\mathrm{tr}}^\star\left(T_{\Phi}(\mathbf z)\right)=m_{\Phi,\mathrm{tr}}^\star\left(\mathbf u\right).
\label{eq:feature_only_pulled_back_target}
\end{equation}
The information-compression bias induced by replacing \(\mathbf Z\) with \(T_{\Phi}(\mathbf Z)\) is then
\begin{equation}
\mathcal B_{\Phi,\mathrm{tr}}^{\mathrm{feat}}:=\mathbb E_{\mathrm{tr}}\left[\left\|m_{\mathrm{tr}}^\star(\mathbf Z)-\bar m_{\Phi,\mathrm{tr}}^\star(\mathbf Z)\right\|_{\mathrm F}^{2}\right].
\label{eq:feature_only_compression_bias}
\end{equation}
This bias is zero only if the information retained by \(T_{\Phi}(\mathbf Z)\) is sufficient to determine the original conditional target. Equivalently, there must exist a measurable function \(q_{\Phi}\) such that
\begin{equation}
m_{\mathrm{tr}}^\star(\mathbf Z)=q_{\Phi}\left(T_{\Phi}(\mathbf Z)\right)
\label{eq:feature_only_sufficiency}
\end{equation}
almost surely under the training distribution. Otherwise, \(\mathcal B_{\Phi,\mathrm{tr}}^{\mathrm{feat}}>0\).

For any predictor \(g(\mathbf U)\), the prediction error can be expanded as
\begin{equation}
\begin{aligned}
\mathbf Y-g(\mathbf U)
&=\mathbf Y-m_{\mathrm{tr}}^\star(\mathbf Z)+m_{\mathrm{tr}}^\star(\mathbf Z)-m_{\Phi,\mathrm{tr}}^\star(\mathbf U)
\\&\quad+m_{\Phi,\mathrm{tr}}^\star(\mathbf U)-g(\mathbf U).
\end{aligned}
\label{eq:feature_only_error_expansion}
\end{equation}
The first component is orthogonal in expectation to the remaining two because
\begin{equation}
\mathbb E_{\mathrm{tr}}\left[\mathbf Y-m_{\mathrm{tr}}^\star(\mathbf Z)\mid \mathbf Z \right]=0.
\label{eq:feature_original_residual_zero_mean}
\end{equation}
Moreover,
\begin{equation}
\begin{aligned}
&
\mathbb E_{\mathrm{tr}}\left[m_{\mathrm{tr}}^\star(\mathbf Z)-m_{\Phi,\mathrm{tr}}^\star(\mathbf U)\mid \mathbf U\right]
\\
&=\mathbb E_{\mathrm{tr}}\left[m_{\mathrm{tr}}^\star(\mathbf Z)\mid\mathbf U\right]-m_{\Phi,\mathrm{tr}}^\star(\mathbf U)=
0,
\end{aligned}
\label{eq:feature_compression_zero_mean}
\end{equation}
so the second component is also orthogonal in expectation to the third. Consequently,
\begin{equation}
\begin{aligned}
&
\mathbb E_{\mathrm{tr}} \left[\left\|\mathbf Y-g(\mathbf U)\right\|_{\mathrm F}^{2}\right]=\mathcal U_{\mathrm{tr}}+\mathcal B_{\Phi,\mathrm{tr}}^{\mathrm{feat}}
\\
&\quad+\mathbb E_{\mathrm{tr}}\left[\left\|m_{\Phi,\mathrm{tr}}^\star(\mathbf U)-g(\mathbf U)\right\|_{\mathrm F}^{2}\right].
\end{aligned}
\label{eq:feature_only_risk_decomposition}
\end{equation}
The three terms respectively represent the irreducible conditional uncertainty under the original observation, the additional error caused by compressing \(\mathbf Z\) into \(T_{\Phi}(\mathbf Z)\), and the predictor-dependent error relative to the prior-only conditional target. In particular, even the population-optimal predictor based only on \(\mathbf U\) has minimum risk \(\mathcal U_{\mathrm{tr}}+\mathcal B_{\Phi,\mathrm{tr}}^{\mathrm{feat}}\).

\textit{2) Prior-augmented input:}
Consider instead the augmented representation
\begin{equation}
\mathbf Z_{\Phi}=\left[\mathbf Z,T_{\Phi}(\mathbf Z)\right].
\label{eq:feature_prior_augmented_input}
\end{equation}
Because the original observation is retained and \(T_{\Phi}(\mathbf Z)\) is deterministically generated from it, \(\mathbf Z_{\Phi}\) and \(\mathbf Z\) contain the same population-level information. Therefore,
\begin{equation}
\begin{aligned}
\mathbb E_{\mathrm{tr}}\left[\mathbf Y \mid \mathbf Z, T_{\Phi}(\mathbf Z) \right]&=\mathbb E_{\mathrm{tr}}\left[\mathbf Y\mid\mathbf Z\right]=m_{\mathrm{tr}}^\star(\mathbf Z).
\end{aligned}
\label{eq:feature_augmentation_same_target}
\end{equation}
Thus, retaining the original observations preserves the optimal radiomap target and \(\mathcal U_{\mathrm{tr}}\). Compared with the prior-only input, however, it presents the model with a higher-dimensional input space, so identifying a mapping close to \(m_{\mathrm{tr}}^\star\) may require a more demanding finite-sample search.

\subsubsection{Model-level prior}

Model-level injection embeds propagation knowledge $F$ into how a radiomap predictor processes and combines information. One example is the adoption of CNN-based architectures\cite{radiounet}, which reflects the prior belief that propagation patterns exhibit local spatial dependence and approximate translation equivariance. Another is to decompose signal propagation into distinct components or regimes, learn them through dedicated branches, and aggregate their contributions to form the final prediction\cite{ref17}.

Let \(\mathcal H_{\Phi}\) denote the effective hypothesis class induced by the model-level prior. Because the predictor still operates on \(\mathbf Z\), its training-domain conditional target remains \(m_{\mathrm{tr}}^\star(\mathbf Z)\). The prior instead restricts the functions available for approximating this target. The resulting structural approximation bias is
\begin{equation}
\mathcal B_{\Phi,\mathrm{tr}}^{\mathrm{model}}
:=\inf_{f\in\mathcal H_{\Phi}}\mathbb E_{\mathrm{tr}}\left[\left\|f(\mathbf Z)-m_{\mathrm{tr}}^\star(\mathbf Z)\right\|_{\mathrm F}^{2}\right].
\label{eq:model_prior_train_bias}
\end{equation}
Accordingly, the minimum attainable training-domain population risk is
\begin{equation}
\inf_{f\in\mathcal H_{\Phi}}R_{\mathrm{tr}}(f)=\mathcal U_{\mathrm{tr}}+\mathcal B_{\Phi,\mathrm{tr}}^{\mathrm{model}}.
\label{eq:model_prior_minimum_train_risk}
\end{equation}
The bias is zero when functions in \(\mathcal H_{\Phi}\) can approximate the conditional-mean radiomap mapping \(m_{\mathrm{tr}}^\star\) arbitrarily well under the training distribution. Otherwise, the imposed structure produces an approximation error that remains even with unlimited data and ideal optimization.

\subsubsection{Loss-level prior}

Without loss of generality, we consider radiomap predictors trained with an additional penalty that encourages consistency with physical propagation mechanisms for loss-level injection. Let
\begin{equation}
\Omega_{\Phi,\mathrm{tr}}(f):=\mathbb E_{\mathrm{tr}}\left[\omega_{\Phi}(f;\mathbf Z,\mathbf Y)\right]
\label{eq:population_prior_penalty}
\end{equation}
denote its training-domain population form. Examples include physical-consistency penalties, propagation-guided auxiliary supervision, and smoothness or boundary constraints. The resulting population objective is
\begin{equation}
\mathcal J_{\lambda,\mathrm{tr}}^{\mathrm{pop}}(f)=R_{\mathrm{tr}}(f)+\lambda\Omega_{\Phi,\mathrm{tr}}(f),\qquad f\in\mathcal H,
\label{eq:loss_prior_population_objective}
\end{equation}
where \(\lambda\geq0\) controls the strength of the prior.

The loss-level prior introduces a soft preference when selecting a predictor from \(\mathcal H\). Because \(\mathcal U_{\mathrm{tr}}\) is independent of \(f\), the corresponding population solution satisfies
\begin{equation}
f_{\lambda,\mathrm{tr}}^\star \in \arg\min_{f\in\mathcal H} \left\{\mathbb E_{\mathrm{tr}}\left[\left\|f(\mathbf Z)-m_{\mathrm{tr}}^\star(\mathbf Z)\right\|_{\mathrm F}^{2}\right]+\lambda\Omega_{\Phi,\mathrm{tr}}(f)\right\}.
\label{eq:loss_prior_train_optimum}
\end{equation}

To isolate the bias induced by this preference from the intrinsic approximation limitation of \(\mathcal H\), define
\begin{equation}
\begin{aligned}
\mathcal B_{\lambda,\mathrm{tr}}^{\mathrm{loss}}
&:=\mathbb E_{\mathrm{tr}}\left[\left\|f_{\lambda,\mathrm{tr}}^\star(\mathbf Z)-m_{\mathrm{tr}}^\star(\mathbf Z)\right\|_{\mathrm F}^{2}\right]
\\
&\quad-\inf_{f\in\mathcal H}\mathbb E_{\mathrm{tr}}\left[\left\|f(\mathbf Z)-m_{\mathrm{tr}}^\star(\mathbf Z)\right\|_{\mathrm F}^{2}\right]\geq 0.
\end{aligned}
\label{eq:loss_prior_train_bias}
\end{equation}
This bias is zero when the loss-guided solution remains a best approximation to \(m_{\mathrm{tr}}^\star\) within \(\mathcal H\). In particular, if \(m_{\mathrm{tr}}^\star\in\mathcal H\) and also minimizes \(\Omega_{\Phi,\mathrm{tr}}\), the original population optimum is preserved. Otherwise, an incompatible penalty with fixed nonzero \(\lambda\) may shift the selected solution away from the conditional target, even with unlimited data and ideal optimization.

\textbf{Discussion:} For radiomap blind prediction, propagation priors can facilitate learning by exposing radiation-related features, organizing model components of the predictor, or designing a loss to guide physical consistency during training. Their potential biases arise
through different mechanisms: feature replacement may discard informative observation details, architectural restrictions may limit the representable mapping, and incompatible loss penalties may favor inaccurate propagation relations. Although feature augmentation preserves the conditional-mean target, the higher input dimension may make identifying an effective radiomap mapping more demanding under finite data and limited optimization. In principle, each injection mechanism can avoid population-level bias if its design preserves the information, representational capacity, or loss compatibility needed to recover the conditional-mean radiomap. In practice, however, these designs are chosen using incomplete observations, finite training data, and selected propagation assumptions; their compatibility with the optimal mapping under incomplete observation is generally difficult to ensure. This motivates a predictor design that retains propagation guidance while allowing data-driven correction of its remaining predictable discrepancies.

\section{RadioDecomp: Correctable Propagation-Prior Learning}
\label{sec:radiodecomp}

\subsection{RadioDecomp Formulation}
\label{subsec:radiodecomp_formulation}
For radiomap blind prediction, let \(b_\Phi(\mathbf z)\) denote a prior-guided estimate of the radiomap generated from the observable environment and BS configuration through the injection mechanisms in Section~\ref{sec:prior_under_incomplete_observation}. The target remains the conditional-mean radiomap \(m^\star(\mathbf z)\). The base estimate may depart from this target because of information compression, structural restrictions, incompatible prior losses, or finite-sample learning. To correct these potential deviations while retaining the encoded propagation guidance, we propose \emph{RadioDecomp}, which augments the prior-guided base with a data-driven residual correction:
\begin{equation}
\hat{\mathbf y}=b_{\Phi}(\mathbf z)+r\left(\mathbf z,b_{\Phi}(\mathbf z)\right),
\label{eq:radiodecomp_general}
\end{equation}
The residual branch receives \(\mathbf z\) to exploit observable cues not captured by the base and \(b_{\Phi}(\mathbf z)\) to condition the correction on the prediction being corrected. Although \(b_{\Phi}(\mathbf z)\) adds no new observable information because it is deterministic given \(\mathbf z\), it serves as a propagation-guided prediction anchor. The base thus provides prior guidance, while the residual branch corrects its data-supported discrepancy.

To formalize this correctability, we consider a fixed prior-guided base predictor. Let \(\mathcal H_R\) denote the residual hypothesis class and define \[\mathcal H_{\mathrm{RD}}(b_{\Phi}):=\left\{f:f(\mathbf z)=b_{\Phi}(\mathbf z)+r\left(\mathbf z,b_{\Phi}(\mathbf z)\right),\quad r\in\mathcal H_R\right\}.\]

\begin{proposition}[Residual correctability]
\label{prop:radiodecomp_correctability}
If the zero function belongs to the residual class, \(0\in\mathcal H_R\), then, for any \(d\in\{\mathrm{tr},\mathrm{te}\}\),
\begin{equation}
\begin{aligned}
&\inf_{f\in\mathcal H_{\mathrm{RD}}(b_{\Phi})}\mathbb E_d\left[\left\|f(\mathbf Z)-m_d^\star(\mathbf Z)\right\|_{\mathrm F}^{2}\right]
\\
&\quad\leq \mathbb E_d\left[\left\|b_{\Phi}(\mathbf Z)-m_d^\star(\mathbf Z)\right\|_{\mathrm F}^{2}\right].
\end{aligned}
\label{eq:radiodecomp_no_worse_than_base}
\end{equation}

Moreover, under squared loss, the pointwise population target of the residual branch in domain \(d\) is
\begin{equation}
r_{\Phi,d}^\star\left(\mathbf z,b_{\Phi}(\mathbf z)\right)=m_d^\star(\mathbf z)-b_{\Phi}(\mathbf z).
\label{eq:radiodecomp_ideal_residual}
\end{equation}
\end{proposition}

\begin{proof}
Because \(0\in\mathcal H_R\), choosing \(r\equiv 0\) recovers \(b_{\Phi}\). Hence, \(b_{\Phi}\in\mathcal H_{\mathrm{RD}}(b_{\Phi})\), which proves \eqref{eq:radiodecomp_no_worse_than_base}.

For a fixed base predictor, \(b_{\Phi}(\mathbf Z)\) is deterministic given \(\mathbf Z\). Therefore, for any \(d\in\{\mathrm{tr},\mathrm{te}\}\),
\begin{equation}
\begin{aligned}
&\mathbb E_d\left[\mathbf Y-b_{\Phi}(\mathbf Z)\mid \mathbf Z=\mathbf z\right]=\mathbb E_d\left[\mathbf Y\mid\mathbf Z=\mathbf z\right]-b_{\Phi}(\mathbf z)
\\
&\quad=m_d^\star(\mathbf z)-b_{\Phi}(\mathbf z).
\end{aligned}
\end{equation}
Because the conditional mean minimizes conditional squared loss,
\eqref{eq:radiodecomp_ideal_residual} follows.
\end{proof}

Proposition~\ref{prop:radiodecomp_correctability} shows that the zero residual recovers \(b_{\Phi}\), while the population-optimal residual in domain \(d\) targets \(m_d^\star-b_{\Phi}\). RadioDecomp thus preserves the base prediction while making its predictable domain-specific discrepancy correctable.

\subsection{Representative Instantiations}
\label{subsec:radiodecomp_instantiations}
RadioDecomp is not tied to a specific propagation prior or base architecture. To evaluate its effectiveness under different prior-guided predictors, we develop two instantiations that examine complementary correction needs in radiomap blind prediction. RadioFR examines residual correction under feature augmentation, whereas RadioLSR examines correction under more restrictive propagation-guided prediction. RadioFR jointly processes the available representations to assess correction when feature augmentation does not introduce information-compression bias. RadioLSR instead uses a LoS-Shadow partition and restricted branch inputs, providing a stricter propagation prior setting. These instantiations demonstrate how the general base-plus-residual framework can be applied to distinct forms of prior injection. We first introduce the observable and propagation-prior representations shared by the two instantiations.

\subsubsection{Common Representations and Two-Stage Training} \label{subsec:radiodecomp_common_inputs}

For each radiomap sample, the observable representation contains an environment description and a BS configuration. The environment is represented by a height map $\mathbf G\in\mathbb R^{H\times W}$, while the BS configuration specifies the antenna array, carrier frequency, transmit power, and beamforming parameters. Based on the observable environment and configuration, several deterministic propagation-related representations are constructed. These representations expose physically meaningful relations that would otherwise need to be inferred entirely from finite training samples.

First, a beam map $\mathbf B \in \mathbb R^{H\times W}$ represents the analytically generated LoS beamforming power. Let \(k\in\{1,\ldots,HW\}\) denote a spatial grid point. The beam-map value is defined as \cite{li2026u6gxlmimoradiomapprediction}
\begin{equation}
B_k=\frac{\lambda^2}{(4\pi)^2}P_t\left|\mathbf w^H\mathbf h_k^{\mathrm{LoS}}\right|^2,
\label{eq:rd_beam_map}
\end{equation}
where \(\lambda\) is the wavelength, \(P_t\) is the transmit power, \(\mathbf w\) is the beamforming vector, and \(\mathbf h_k^{\mathrm{LoS}}\) denotes the analytically generated LoS channel at grid point \(k\). The beam map therefore exposes the configuration-dependent direct-coverage pattern.

Second, two directional edge maps are extracted from the height map:
\begin{equation}
\begin{aligned}
\mathbf E^{(x)}
&=\mathcal D_x\left(\mathbf G\right),
\\
\mathbf E^{(y)}&=\mathcal D_y\left(\mathbf G\right),
\end{aligned}
\label{eq:rd_edge_maps}
\end{equation}
where \(\mathcal D_x(\cdot)\) and \(\mathcal D_y(\cdot)\) denote directional edge operators along the two spatial axes. The edge maps expose geometric boundaries that may affect blockage, reflection, and diffraction.

Third, a blockage score map $\mathbf S \in \mathbb R^{H\times W}$ is constructed from the observable geometry and BS position. For grid point \(k\), the BS-grid line segment is sampled at \(M\) intermediate locations:
\begin{equation}
S_k=\frac{1}{M} \sum_{m=1}^{M}\left[h_{\mathrm{bld}}\left(\mathbf p_{\mathrm{BS}}+t_m\left(\mathbf p_k-\mathbf p_{\mathrm{BS}}\right)\right)-h_{\mathrm{ray}}\left(t_m\right)\right]_+,
\label{eq:rd_blockage_score}
\end{equation}
where \(\mathbf p_{\mathrm{BS}}\) and \(\mathbf p_k\) denote the horizontal positions of the BS and grid point \(k\), respectively, \(\{t_m\}_{m=1}^{M}\subset(0,1)\) are interpolation factors, \(h_{\mathrm{bld}}(\cdot)\) denotes the building height, \(h_{\mathrm{ray}}(t_m)\) denotes the height of the direct BS-grid segment, and \([a]_+=\max(a,0)\). The blockage score is clipped and normalized before being supplied to the network.

Prediction is performed only over the valid non-building region. Let $\mathbf M_{\mathrm{valid}} \in \{0,1\}^{H\times W}$ denote the corresponding binary mask. The propagation-feature representation is written as
\begin{equation}
\mathbf X_{\Phi}=\left[\mathbf B,\mathbf G,\mathbf E^{(x)},\mathbf E^{(y)},\mathbf S\right].
\label{eq:rd_common_feature_tensor}
\end{equation}

The beam map, directional edge maps, and blockage score are deterministic transformations of the available environment and configuration information. Therefore, they do not constitute genuine observation enrichment in the sense of Proposition~\ref{prop:observation_enrichment}. Instead, they reorganize the original observation into propagation-related representations and serve as feature-level propagation priors. RadioFR and RadioLSR are constructed from the same underlying observable environment and configuration information. However, they organize and supply the above representations differently according to their respective prior-guided base structures.

For subsequent training, let \(\mathbf y\) denote the normalized ground-truth radiomap. For a prediction \(\mathbf a\), target \(\mathbf b\), and binary mask \(\mathbf M\), define the masked squared \(L_2\) loss as
\begin{equation}
\ell_{\mathrm{mask}}^{(2)}\left(\mathbf a,\mathbf b;\mathbf M\right)=\frac{\left\|\left(\mathbf a-\mathbf b\right)\odot \mathbf M\right\|_{\mathrm F}^{2}}{\max\left(\left\|\mathbf M \right\|_1,1\right)}.
\label{eq:radiodecomp_masked_l2}
\end{equation}

Both instantiations follow a two-stage training procedure. The prior-guided base is first trained using its model-specific objective and then frozen. Only the residual branch is subsequently optimized against the difference between the training radiomap and the frozen base prediction under the masked squared \(L_2\) loss. All prediction branches in RadioFR and RadioLSR use the common U-Net backbone summarized in Table~\ref{tab:radiodecomp_unet_backbone}, where \(C_b\) denotes the initial channel width of each branch.

\subsubsection{RadioFR (Feature-Residual Prediction)}
\label{subsec:radiofr}
RadioFR realizes RadioDecomp with a monolithic network that jointly uses all propagation-related features
to predict the radiomap. These feature representation in \eqref{eq:rd_common_feature_tensor} is directly supplied to a U-Net:
\begin{equation}
b_{\mathrm F}\left(\mathbf z\right)=f_{\mathrm F}\left(\mathbf X_{\Phi}\right)\odot \mathbf M_{\mathrm{valid}},
\label{eq:radiofr_base}
\end{equation}
Here, \(f_{\mathrm F}\), a monolithic U-Net, jointly processes all input channels. Relative to the height-map and beam-map inputs, adding edge and blockage features preserves the available information and does not introduce compression bias. Nevertheless, finite capacity, limited training data, and imperfect optimization may leave predictable radiomap errors. RadioFR therefore tests whether residual correction remains useful in this general feature-augmentation setting. The residual branch receives the same features together with the base radiomap:

\begin{equation}
\mathbf X_{\mathrm{R,F}}=\left[\mathbf B,\mathbf G,\mathbf E^{(x)},\mathbf E^{(y)},\mathbf S,b_{\mathrm F}\left(\mathbf z
\right)\right].
\label{eq:radiofr_residual_input}
\end{equation}

The residual prediction is
\begin{equation}
\hat{\mathbf r}_{\mathrm F}=g_{\mathrm F}\left(\mathbf X_{\mathrm{R,F}}\right)\odot \mathbf M_{\mathrm{valid}},
\label{eq:radiofr_residual}
\end{equation}
where \(g_{\mathrm F}\) is implemented using another U-Net.

The final RadioFR prediction is
\begin{equation}
\hat{\mathbf y}_{\mathrm F}=b_{\mathrm F}\left(\mathbf z\right)+\hat{\mathbf r}_{\mathrm F}.
\label{eq:radiofr_final}
\end{equation}

In the first training stage, the feature-guided base is optimized using
\begin{equation}
\mathcal L_{\mathrm{base,F}}=\ell_{\mathrm{mask}}^{(2)}\left(b_{\mathrm F}(\mathbf z),\mathbf y;\mathbf M_{\mathrm{valid}}\right).
\label{eq:radiofr_base_loss}
\end{equation}

In the second stage, the trained base is frozen, and the residual target is defined as
\begin{equation}
\mathbf y_{\mathrm{res,F}} =\mathbf y-b_{\mathrm F}(\mathbf z).
\label{eq:radiofr_residual_target}
\end{equation}
Only the residual branch is then optimized using
\begin{equation}
\mathcal L_{\mathrm{res,F}}=\ell_{\mathrm{mask}}^{(2)}\left(\hat{\mathbf r}_{\mathrm F},\mathbf y_{\mathrm{res,F}};\mathbf M_{\mathrm{valid}}\right).
\label{eq:radiofr_residual_loss}
\end{equation}

\subsubsection{RadioLSR: LoS-Shadow-Residual Prediction}
\label{subsec:radiolsr}
RadioLSR examines a more restrictive propagation-guided base for radiomap blind prediction. It partitions the prediction region using the blockage score and assigns different inputs to the two branches: the beam map for the LoS region, and the beam map together with the blockage score for the Shadow region. This design emphasizes direct-coverage and blockage trends while limiting the geometric information available to the base prediction. Based on the blockage score, the LoS and Shadow masks are defined as
\begin{equation}
\begin{aligned}
\mathbf M_{\mathrm{LoS}}
&=\mathbb I\left(\mathbf S\le \epsilon \right) \odot \mathbf M_{\mathrm{valid}},
\\
\mathbf M_{\mathrm{Shd}}&=\mathbb I\left(\mathbf S> \epsilon \right) \odot \mathbf M_{\mathrm{valid}},
\end{aligned}
\label{eq:radiolsr_masks}
\end{equation}
where \(\epsilon\) is the blockage-score threshold. The two masks satisfy
\begin{equation}
\mathbf M_{\mathrm{LoS}}+\mathbf M_{\mathrm{Shd}}=\mathbf M_{\mathrm{valid}}.
\label{eq:radiolsr_mask_partition}
\end{equation}

The LoS branch uses the beam map to estimate the direct-coverage trend:
\begin{equation}
\hat{\mathbf y}_{\mathrm{LoS}}=f_{\mathrm{LoS}}\left(\mathbf B
\right)\odot\mathbf M_{\mathrm{LoS}},
\label{eq:radiolsr_los_branch}
\end{equation}
where \(f_{\mathrm{LoS}}\) is implemented using a U-Net.

The Shadow branch jointly uses the beam map and blockage score:
\begin{equation}
\hat{\mathbf y}_{\mathrm{Shd}}=f_{\mathrm{Shd}}\left(\left[\mathbf B,\mathbf S\right]\right)\odot \mathbf M_{\mathrm{Shd}},
\label{eq:radiolsr_shadow_branch}
\end{equation}
where \(f_{\mathrm{Shd}}\) is another U-Net.

The LoS-Shadow structured base is then given by
\begin{equation}
\begin{aligned}
b_{\mathrm{LS}}\left(\mathbf z\right)
&=\hat{\mathbf y}_{\mathrm{LoS}}+\hat{\mathbf y}_{\mathrm{Shd}}
\\
&=f_{\mathrm{LoS}}\left(\mathbf B\right)\odot \mathbf M_{\mathrm{LoS}}+f_{\mathrm{Shd}}\left(\left[\mathbf B,\mathbf S\right]\right)\odot \mathbf M_{\mathrm{Shd}}.
\end{aligned}
\label{eq:radiolsr_base}
\end{equation}

The two branches are supervised by the total radiomap within their respective regions. However, their restricted inputs may miss geometry-dependent coverage variations, including those associated with reflection and diffraction. The residual branch therefore receives the height and edge maps alongside the beam map and base prediction to recover predictable spatial details unavailable to the base:
\begin{equation}
\mathbf X_{\mathrm{R,LS}}=\left[\mathbf B,\mathbf G,\mathbf E^{(x)},\mathbf E^{(y)},b_{\mathrm{LS}}(\mathbf z)\right].
\label{eq:radiolsr_residual_input}
\end{equation}

The residual prediction is
\begin{equation}
\hat{\mathbf r}_{\mathrm{LS}}=g_{\mathrm{LS}}\left(\mathbf X_{\mathrm{R,LS}}\right)\odot \mathbf M_{\mathrm{valid}},
\label{eq:radiolsr_residual}
\end{equation}
where \(g_{\mathrm{LS}}\) is implemented using a U-Net.

The final RadioLSR prediction is
\begin{equation}
\hat{\mathbf y}_{\mathrm{LS}}=b_{\mathrm{LS}}\left(\mathbf z \right)+\hat{\mathbf r}_{\mathrm{LS}}.
\label{eq:radiolsr_final}
\end{equation}

In the first training stage, the LoS and Shadow branches are optimized
over their corresponding regions using
\begin{equation}
\begin{aligned}
\mathcal L_{\mathrm{base,LS}}
={}&
\lambda_{\mathrm{LoS}}
\ell_{\mathrm{mask}}^{(2)}
\left(
\hat{\mathbf y}_{\mathrm{LoS}},
\mathbf y;
\mathbf M_{\mathrm{LoS}}
\right)
\\
&+
\lambda_{\mathrm{Shd}}
\ell_{\mathrm{mask}}^{(2)}
\left(
\hat{\mathbf y}_{\mathrm{Shd}},
\mathbf y;
\mathbf M_{\mathrm{Shd}}
\right).
\end{aligned}
\label{eq:radiolsr_base_loss}
\end{equation}
where \(\lambda_{\mathrm{LoS}}\geq 0\) and \(\lambda_{\mathrm{Shd}}\geq 0\) weight the LoS- and Shadow-region objectives, respectively.

In the second stage, the trained base branches are frozen, and only \(g_{\mathrm{LS}}\) is optimized using
\begin{equation}
\mathcal L_{\mathrm{res,LS}}=\ell_{\mathrm{mask}}^{(2)}\left(\hat{\mathbf r}_{\mathrm{LS}}, \mathbf y-b_{\mathrm{LS}}(\mathbf z);\mathbf M_{\mathrm{valid}}\right).
\label{eq:radiolsr_residual_loss}
\end{equation}

\begin{table}[t]
\vspace{4pt}
\centering
\caption{Unified U-Net backbone used by the prediction branches.}
\label{tab:radiodecomp_unet_backbone}
\setlength{\tabcolsep}{4pt}
\begin{tabular}{c|c}
\toprule
\textbf{Stage}
&
\textbf{Specification}
\\
\midrule
Encoder
&
DoubleConv \((C_b)\), then
3\(\times\)[MaxPool + DoubleConv]
\\
&
Channels:
\(2C_b,\,4C_b,\,8C_b\)
\\
Bottleneck
&
DoubleConv \((8C_b)\)
\\
Decoder
&
3\(\times\)[Up + skip concat + DoubleConv]
\\
&
Channels:
\(4C_b,\,2C_b,\,C_b\)
\\
Output
&
\(1\times1\) convolution with one output channel
\\
\bottomrule
\end{tabular}
\end{table}

\section{Experiments}
\label{sec:experiments}
In the experiments, we validate the preceding analysis and evaluate the effectiveness of RadioDecomp under different training and evaluation conditions.

\subsection{Experimental Setup}
\label{subsec:experimental_setup}

\subsubsection{Dataset and Evaluation Protocols}

Experiments are conducted on the U6G XL-MIMO Radiomap dataset \cite{li2026u6gxlmimoradiomapprediction}, which contains \(78{,}400\) radiomaps generated from \(800\) urban scenes and \(98\) BS configurations. Each sample contains an observable environment representation, a BS configuration representation, and a \(128\times128\) received-power radiomap. Unless otherwise specified, the dataset is divided using a \(6/1/3\) training/validation/test ratio with random seed \(42\). Three evaluation protocols are considered:

\begin{itemize}
    \item \emph{Random}: all radiomap samples are randomly divided into training, validation, and test sets.

    \item \emph{Cross-config}: the BS configurations used for testing are disjoint from those used for training.

    \item \emph{Cross-env}: the urban scenes used for testing are disjoint from those used for training.
\end{itemize}

\subsubsection{Compared Predictors}

The following controlled predictors are considered:

\begin{itemize}
\item \emph{Vanilla}: a standalone U-Net that directly predicts the radiomap from the observable environment and BS configuration. The environment is represented by the height map, while the BS configuration is represented by a binary \(3\)-dB beam mask.

\item \emph{BaseFR}: the feature-guided monolithic base defined in Section~\ref{subsec:radiofr}, without residual correction.

\item \emph{RadioFR}: the feature-guided monolithic instantiation defined in Section~\ref{subsec:radiofr}.

\item \emph{BaseLSR}: the LoS-Shadow structured base defined in Section~\ref{subsec:radiolsr}, without residual correction.

\item \emph{RadioLSR}: the complete LoS-Shadow-Residual instantiation defined in Section~\ref{subsec:radiolsr}.
\end{itemize}

Unless otherwise specified, all experiments use the following initial channel widths. BaseFR uses \(C_b=32\), while the base and residual branches of RadioFR both use \(C_b=32\). For BaseLSR and RadioLSR, the LoS and Shadow branches each use \(C_b=16\); the residual branch of RadioLSR uses \(C_b=32\).

\subsubsection{Training and Metrics}

All predictors are optimized using AdamW with a batch size of \(32\) and an initial learning rate of \(10^{-3}\). All training objectives use the masked squared \(L_2\) loss defined in \eqref{eq:radiodecomp_masked_l2}. RadioFR and RadioLSR are trained in two stages: the base predictor is trained first and then frozen, after which only the residual branch is optimized. Unless otherwise specified, each training stage uses \(60\) epochs, and the checkpoint with the lowest validation loss is retained. For the RadioLSR base, \(\lambda_{\mathrm{LoS}}=\lambda_{\mathrm{Shd}}=0.5\).

Performance is evaluated over the valid non-building region. Training and test Root Mean Squared Error (RMSE) are reported to measure squared-error performance, while test Mean Absolute Error (MAE) is additionally provided as a complementary measure of average prediction error that is less sensitive to large local errors.

\subsection{Propagation-Prior Guidance and Residual Correction}
\label{subsec:prior_residual_effects}

We first evaluate the effects of propagation-prior injection and residual correctability under the three evaluation protocols. Table~\ref{tab:prior_residual_main} compares the considered predictors.

\begin{table*}[!t]
\centering
\caption{Training RMSE and test RMSE/MAE (dB) under the \(6/1/3\)
random, cross-config, and cross-env protocols.}
\label{tab:prior_residual_main}
\scriptsize
\setlength{\tabcolsep}{2.5pt}
\begin{tabular}{l|ccc|ccc|ccc}
\toprule
\multirow{2}{*}{\textbf{Method}}
&
\multicolumn{3}{c|}{\textbf{Random}}
&
\multicolumn{3}{c|}{\textbf{Cross-config}}
&
\multicolumn{3}{c}{\textbf{Cross-env}}
\\
\cline{2-10}
&
\textbf{Train RMSE}
&
\textbf{Test RMSE}
&
\textbf{Test MAE}
&
\textbf{Train RMSE}
&
\textbf{Test RMSE}
&
\textbf{Test MAE}
&
\textbf{Train RMSE}
&
\textbf{Test RMSE}
&
\textbf{Test MAE}
\\
\midrule
Vanilla
& 5.7382  & 5.9414 & 4.0935
& 5.7474  & 11.3364 & 8.7591
& 8.5897  & 11.6593 & 8.3522
\\
\midrule
BaseFR
& 3.4333  & 3.6075 & 2.1264
& 3.7145  & 6.4390 & 4.6388
& 6.1419  & 8.1933 & 4.7766
\\
RadioFR
& 3.3841  & 3.5588 & 2.0856
& 3.0630  & \textbf{5.9706} & \textbf{4.2466}
& \textbf{4.8478}  & \textbf{7.8755} & \textbf{4.4046}
\\
\midrule
BaseLSR
& 7.2615  & 7.2742 & 3.7781
& 7.2774  & 11.0581 & 8.2632
& 8.4982  & 10.0741 & 5.3493
\\
RadioLSR
& \textbf{3.2009}  & \textbf{3.4356} & \textbf{1.8905}
& \textbf{3.0369}  & 9.1946 & 7.1526
& 5.9953  & 8.0501 & 4.4138
\\
\bottomrule
\end{tabular}
\end{table*}

Under the random protocol, the training and test RMSE values remain
close for all predictors, with gaps below \(0.24\) dB. This confirms that random splitting produces only a small train-test gap. In contrast, the gaps become substantially larger under the cross-config and cross-env protocols. Vanilla is consistently outperformed by BaseFR, RadioFR, and RadioLSR in both test RMSE and MAE under all three protocols. In particular, compared with Vanilla, BaseFR reduces test RMSE by \(2.3339\), \(4.8974\), and \(3.4660\) dB under the random, cross-config, and cross-env protocols, respectively. Since BaseFR does not include residual correction, these improvements demonstrate the benefit of injecting propagation-related representations. BaseLSR does not provide uniformly better standalone performance because its LoS-Shadow structure deliberately models only selected propagation effects, leaving the remaining discrepancy for residual correction. 

Consistent with this design, residual correction improves both prior-guided bases under all three protocols. Relative to BaseFR, RadioFR reduces test RMSE by \(0.0487\), \(0.4684\), and \(0.3178\) dB. The correction is more pronounced for BaseLSR: RadioLSR reduces test RMSE by \(3.8386\), \(1.8635\), and \(2.0240\) dB. The consistent improvements over two structurally different bases support the effectiveness of RadioDecomp in correcting their remaining predictable discrepancies.

Fig.~\ref{fig:qualitative_comparison} provides representative spatial comparisons under the two shifted protocols. RadioFR refines the local coverage details captured by BaseFR and reduces the corresponding localized errors. The correction is particularly pronounced from BaseLSR to RadioLSR in the cross-config setting: the residual branch captures additional spatial variations not represented by the LoS-Shadow base, which may arise from reflection and other complex propagation interactions. These results visually support the corrective role of RadioDecomp.

\begin{figure*}[!t]
\centering
\includegraphics[width=0.99\textwidth]
{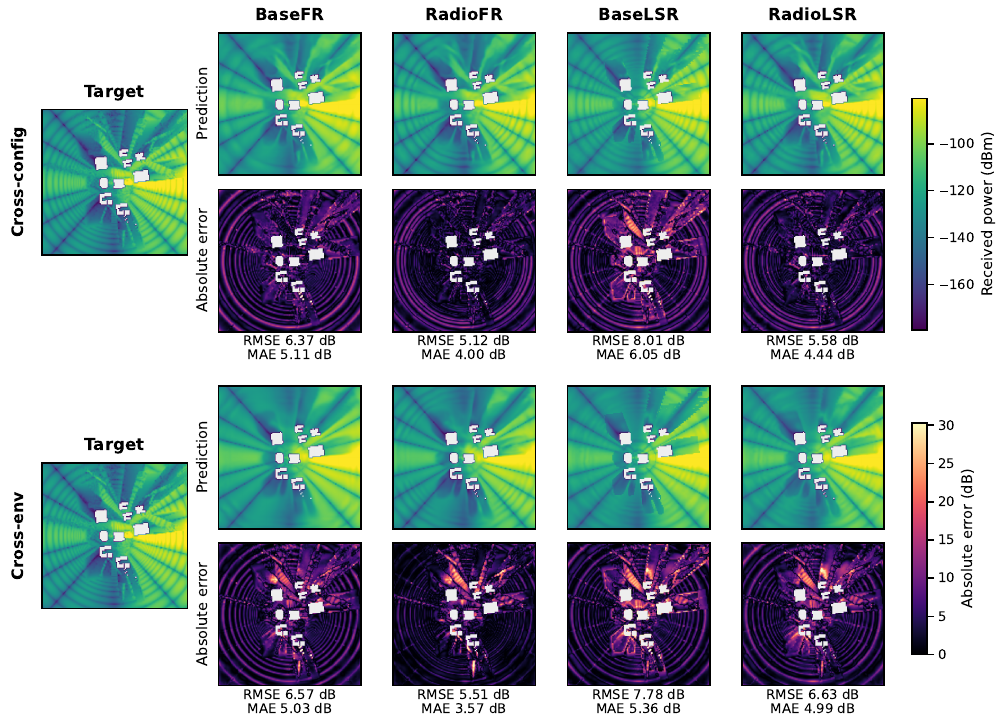}
\caption{Qualitative comparison on representative cross-config and cross-env test samples, with radiomaps and error maps using common received-power and error scales, respectively, white regions denoting masked buildings, and per-sample RMSE and MAE reported below each error map.}
\label{fig:qualitative_comparison}
\end{figure*}

\subsection{Propagation-Prior Learning under Extremely Limited Supervision}
\label{subsec:limited_sample_prior}

In practical deployments, acquiring sufficient radiomaps for a target area can be costly. We therefore restrict the training set to only \(10\) radiomaps under the random protocol to evaluate RadioDecomp under extremely limited supervision. Each training stage uses \(200\) epochs in this experiment. The results are reported in Table~\ref{tab:tiny_sample_prior}.

\begin{table}[t]
\centering
\caption{Prediction performance with only \(10\) training radiomaps
under the random protocol.}
\label{tab:tiny_sample_prior}
\small
\setlength{\tabcolsep}{4.5pt}
\begin{tabular}{lccc}
\toprule
\textbf{Method}
&
\textbf{Train RMSE}
&
\textbf{Test RMSE}
&
\textbf{Test MAE}
\\
\midrule

Vanilla
& 14.9871
& 17.9310
& 13.6009
\\

BaseFR
& 12.3616
& 15.6376
& 11.0481
\\

BaseLSR
& 11.6288
& 20.8432
& 11.8843
\\

RadioFR
& 11.2497
& \textbf{14.4285}
& 10.4248
\\

RadioLSR
& \textbf{7.9732}
& 19.6369
& \textbf{9.3178}
\\

\bottomrule
\end{tabular}
\end{table}

Compared with the standard random protocol, the ten-sample setting produces a substantially larger train-test gap. The test RMSE exceeds the training RMSE by at least \(2.9439\) dB for all predictors, whereas the corresponding gaps (random) in Table~\ref{tab:prior_residual_main} are below \(0.24\) dB. Vanilla exhibits the smallest train-test RMSE gap of \(2.9439\) dB. This may partly result from its simpler input representation and weaker training fit, as it also has the highest training RMSE. Moreover, BaseFR reduces the test RMSE from \(17.9310\) to \(15.6376\) dB and the test MAE from \(13.6009\) to \(11.0481\) dB, demonstrating that propagation-related representations remain beneficial when labeled radiomaps are scarce.

Residual correction improves both prior-guided bases. Relative to BaseFR, RadioFR reduces the training RMSE by \(1.1119\) dB and the test RMSE by \(1.2091\) dB, achieving the lowest test RMSE of \(14.4285\) dB. Relative to BaseLSR, RadioLSR reduces the training RMSE by \(3.6556\) dB, the test RMSE by \(1.2063\) dB, and the test MAE by \(2.5665\) dB. These consistent reductions support the effectiveness of RadioDecomp under extremely limited supervision.

RadioLSR achieves the lowest test MAE of \(9.3178\) dB, although its test RMSE remains relatively high. Its LoS-Shadow partition may make the dominant propagation patterns easier to fit over most locations, thereby reducing average absolute error, while leaving a small number of large errors in complex or imperfectly partitioned regions. These localized errors receive substantially greater weight in RMSE, which explains the divergence between the two metrics.

\subsection{Residual Correction across Feature-Guided Base Capacities}
\label{subsec:capacity_control}

Although feature concatenation does not alter the population conditional target, a base predictor learned from finite samples may still retain an optimisation discrepancy. We examine this effect by varying the initial channel width of BaseFR over \(C_b\in\{32,48,64,96\}\). For each width, RadioFR uses the corresponding frozen BaseFR and a residual branch with a fixed width of \(32\). \footnote{Fixing the residual capacity provides a controlled comparison across base widths; a wider residual branch may provide additional correction but is not considered here.} Table~\ref{tab:capacity_control} reports the results under the three evaluation protocols.

Under the random and cross-config protocols, residual correction is more beneficial at smaller or intermediate base widths. The largest RMSE reductions are \(0.1707\) dB under the random protocol and \(0.4684\) dB under cross-config. At \(C_b=96\), the absolute RMSE differences are only \(0.0031\) and \(0.0038\) dB under the random and cross-config protocols, respectively, with RadioFR being marginally worse in the latter case. These results suggest that a sufficiently expressive BaseFR can capture most of the predictable relation in these settings. In particular, the analytically constructed beam map directly represents configuration-dependent radiation, which may limit the remaining correctable discrepancy under cross-config evaluation.

A different pattern appears under cross-env, where RadioFR reduces test RMSE at every width from \(0.2245\) to \(0.3315\) dB. Cross-environment prediction involves unseen spatial layouts, for which the relation between the observable geometric features and the radiomap is less completely covered by the training samples. The residual branch therefore remains useful for correcting finite-sample errors relative to the feature-guided base. RadioFR-64 also achieves a lower cross-env RMSE than BaseFR-96 (\(7.6678\) versus \(8.0248\) dB) despite using fewer parameters, indicating that the gain is not attributable solely to increased model capacity.

\begin{table*}[!t]
\centering
\caption{Residual correction across feature-guided base capacities under the three evaluation protocols, with all RadioFR residual branches using \(C_b=32\) and bold indicating the lower error within each BaseFR-RadioFR pair.
}
\label{tab:capacity_control}
\scriptsize
\setlength{\tabcolsep}{2.4pt}
\begin{tabular}{l|c|ccc|ccc|ccc}
\toprule
\multirow{2}{*}{\textbf{Method}}
&
\multirow{2}{*}{\textbf{Params}}
&
\multicolumn{3}{c|}{\textbf{Random}}
&
\multicolumn{3}{c|}{\textbf{Cross-config}}
&
\multicolumn{3}{c}{\textbf{Cross-env}}
\\
\cline{3-11}
&
&
\textbf{Train RMSE}
&
\textbf{Test RMSE}
&
\textbf{Test MAE}
&
\textbf{Train RMSE}
&
\textbf{Test RMSE}
&
\textbf{Test MAE}
&
\textbf{Train RMSE}
&
\textbf{Test RMSE}
&
\textbf{Test MAE}
\\
\midrule
BaseFR-32
& 3.1295M
& 3.4333 & 3.6075 & 2.1264
& 3.7145 & 6.4390 & 4.6388
& 6.1419 & 8.1933 & 4.7766
\\

BaseFR-48
& 7.0375M
& 2.9885 & 3.2183 & 1.9032
& 2.7451 & 5.5745 & 3.9754
& 4.6841 & 8.0554 & 4.4492
\\

BaseFR-64
& 12.5075M
& 2.5282 & 2.8246 & 1.6136
& 2.5515 & 5.6783 & 4.0489
& 4.9321 & 7.9993 & 4.4656
\\

BaseFR-96
& 28.1340M
& {2.5231} & {2.8096} & {1.6132}
& {2.1523} & \textbf{5.7090} & {4.0976}
& {4.1787} & {8.0248} & {4.2922}
\\

\midrule
RadioFR-32
& 6.2594M
& \textbf{3.3841} & \textbf{3.5588} & \textbf{2.0856}
& \textbf{3.0630} & \textbf{5.9706} & \textbf{4.2466}
& \textbf{4.8478} & \textbf{7.8755} & \textbf{4.4046}
\\

RadioFR-48
& 10.1673M
& \textbf{2.7781} & \textbf{3.0476} & \textbf{1.7520}
& \textbf{2.7298} & \textbf{5.5734} & \textbf{3.9512}
& \textbf{4.1774} & \textbf{7.7778} & \textbf{4.2241}
\\

RadioFR-64
& 15.6373M
& \textbf{2.4831} & \textbf{2.7865} & \textbf{1.5862}
& \textbf{2.5248} & \textbf{5.5063} & \textbf{3.8939}
& \textbf{4.2651} & \textbf{7.6678} & \textbf{4.1429}
\\

RadioFR-96
& 31.2638M
& \textbf{2.5185} & \textbf{2.8065} & \textbf{1.6042}
& \textbf{2.1504} & 5.7128 & \textbf{4.0915}
& \textbf{3.9909} & \textbf{7.8003} & \textbf{4.2199}
\\
\bottomrule
\end{tabular}
\end{table*}

\subsection{Training-Support Coverage under a Fixed Test Domain}
\label{subsec:fixed_test_coverage}

Corollary~\ref{cor:risk_gap} motivates a practical data-acquisition question: when the objective is to minimize prediction error over a prescribed target domain, should training data prioritize relevant environments, relevant BS configurations, or both? To isolate how this choice affects the train-test gap, we keep the test set fixed. Consequently, the test observable distribution and conditional-uncertainty term remain unchanged, and performance differences mainly reflect how the selected training data affect the learned predictor over the same target domain.

We partition the \(800\) environments into \(\mathcal E_{\mathrm{te}}\) and \(\mathcal E_{\mathrm{other}}\), each containing \(400\) environments, and the \(98\) configurations into \(\mathcal C_{\mathrm{te}}\) and \(\mathcal C_{\mathrm{other}}\), each containing \(49\) configurations. The fixed test set is $\mathcal D_{\mathrm{te}}=\mathcal E_{\mathrm{te}}\times\mathcal C_{\mathrm{te}}$, which contains \(19{,}600\) radiomaps. We consider four fixed-budget training policies:
\begin{itemize}
\item \emph{Both Covered}:
\((\mathcal E_{\mathrm{te}}\times\mathcal C_{\mathrm{other}})\cup (\mathcal E_{\mathrm{other}}\times\mathcal C_{\mathrm{te}})\);

\item \emph{Environment Only}:
\(\mathcal E_{\mathrm{te}}\times\mathcal C_{\mathrm{other}}\);

\item \emph{Configuration Only}:
\(\mathcal E_{\mathrm{other}}\times\mathcal C_{\mathrm{te}}\);

\item \emph{Neither Covered}:
\(\mathcal E_{\mathrm{other}}\times\mathcal C_{\mathrm{other}}\).
\end{itemize}

Each policy uses \(17{,}600\) training and \(2{,}000\) validation radiomaps. We additionally include \emph{Both Covered (Full)}, which uses the complete leakage-free Both Covered pool, with \(37{,}200\) training and \(2{,}000\) validation samples. No exact environment-configuration pair in the fixed test set is used for training or validation. The RadioLSR architecture and optimization settings are held constant.

\begin{table}[!t]
\centering
\caption{Fixed-test performance of RadioLSR under different policies.}
\label{tab:fixed_test_coverage}
\small
\setlength{\tabcolsep}{4.0pt}

\begin{tabular}{l|ccc}
\toprule

\textbf{Training Policy}
&
\textbf{Train RMSE}
&
\textbf{Test RMSE}
&
\textbf{Test MAE}
\\

\midrule

Both Covered
&
3.4032
&
4.6924
&
2.5305

\\

Environment Only
&
3.1738
&
5.7332
&
3.8413

\\

Configuration Only
&
\textbf{3.0217}
&
9.3269
&
4.4443

\\

Neither Covered
&
3.0554
&
10.4874
&
6.1990

\\

Both Covered (Full)
&
3.1288
&
\textbf{4.4540}
&
\textbf{2.3567}

\\

\bottomrule
\end{tabular}
\end{table}

As shown in Table~\ref{tab:fixed_test_coverage}, the four policies obtain similar training RMSE values, ranging only from \(3.0217\) to \(3.4032\) dB under the common \(17{,}600\)-sample budget, but their fixed-test performance differs substantially. Neither Covered produces the highest test RMSE and MAE of \(10.4874\) and \(6.1990\) dB. Relative to this policy, Configuration Only reduces the two errors by \(1.1605\) and \(1.7547\) dB, whereas Environment Only reduces them by \(4.7542\) and \(2.3577\) dB. Thus, when only one factor can be covered, acquiring data from test-relevant environments is more beneficial than acquiring data with test-relevant configurations in the current dataset utilising RadioLSR. Covering both factors yields complementary gains and achieves the best fixed-budget results, with a test RMSE of \(4.6924\) dB and MAE of \(2.5305\) dB. These differences, with comparable training errors and no pair-level leakage, show that the train-test gap depends strongly on where the training samples are collected.

Expanding Both Covered to the complete candidate pool more than doubles the training-set size, from \(17{,}600\) to \(37{,}200\), but provides smaller additional reductions of \(0.2384\) dB in RMSE and \(0.1738\) dB in MAE. The results suggest a practical acquisition priority: training data should first cover both the target-relevant environments and configurations, with particular attention to environment coverage; increasing sampling density within the covered support provides a further but smaller benefit.

\subsection{Effect of Observable-Representation Coarsening}
\label{subsec:observation_coarsening}

Proposition~\ref{prop:observation_enrichment} shows that enriching the observation with additional predictive information can reduce irreducible conditional uncertainty. We examine this effect in the reverse direction by progressively coarsening the observable environment and BS-configuration representations. The same simple monolithic U-Net, optimization strategy, and random \(6/1/3\) split across all settings are used. Training RMSE is treated as the primary indicator of information loss, while test RMSE is reported as a practical reference because input coarsening may also change the train-test distribution mismatch.

We consider two BS-configuration representations: the continuous beam map in Section~\ref{subsec:radiodecomp_common_inputs} and its binary \(3\)-dB high-gain-region mask. The environment is represented by either the complete height map or a coarsened version with \(10\%\) of the height observations removed. Table~\ref{tab:observation_coarsening} reports the four resulting input combinations.

\begin{table}[!t]
\centering
\caption{RMSE under different beam and height-map representations.}
\label{tab:observation_coarsening}
\small
\setlength{\tabcolsep}{3.5pt}
\begin{tabular}{ll|cc}
\toprule
\textbf{Beam Input}
&
\textbf{Height Input}
&
\textbf{Train RMSE}
&
\textbf{Test RMSE}
\\
\midrule

Continuous
&
Complete
&
\textbf{3.4488}
&
\textbf{3.6185}
\\

Continuous
&
\(10\%\) missing
&
3.5844
&
3.7440
\\

\(3\)-dB mask
&
Complete
&
5.7382
&
5.9414
\\

\(3\)-dB mask
&
\(10\%\) missing
&
5.7254
&
5.9323
\\

\bottomrule
\end{tabular}
\end{table}

The continuous beam representation provides the dominant improvement. Relative to the \(3\)-dB mask, it reduces training RMSE by \(2.2894\) dB with the complete height map and by \(2.1410\) dB under \(10\%\) height missingness. In comparison, removing \(10\%\) of the height observations increases training RMSE by only \(0.1356\) dB with the full beam map and changes it by a negligible \(0.0128\) dB in the opposite direction with the \(3\)-dB mask. Under the controlled training setting, the substantial reductions obtained from the continuous beam map are consistent with the lower conditional uncertainty enabled by richer observable information. The test results largely follow the same dominant trend. The only exception is that the fourth setting outperforms the third by \(0.0091\) dB. This minor reversal may result from changes in the induced train-test distribution mismatch or finite-sample optimization variation and does not indicate that removing height information lowers conditional uncertainty.

\section{Conclusion}
\label{sec:conclusion}

This paper investigated deterministic radiomap blind prediction under incomplete observation. Under squared loss, we identified its population-optimal target as the conditional-mean radiomap. Based on this target, we established an exact risk decomposition into reducible predictor approximation and irreducible uncertainty induced by hidden propagation factors. Building on this decomposition, we characterized the train-test risk gap through domain-dependent changes in both terms and proved that informative observation enrichment reduces the uncertainty term.

With the observation fixed, predictor design can further reduce approximation error. We therefore examined propagation priors as physically grounded guidance for this reduction, while showing that their implementations may bias the attainable predictor. To retain this guidance while correcting such deviations, we proposed RadioDecomp, which augments a prior-guided base with residual refinement. Its feature-guided and LoS-Shadow structured instantiations, RadioFR and RadioLSR, improved upon their corresponding bases under random, distribution-shifted, and limited-data settings. Controlled studies of model capacity, training-support coverage, and observation coarsening further supported the analysis. Overall, observation enrichment and correctable propagation-prior learning provide complementary routes to reducing irreducible uncertainty and predictor-approximation error, respectively.

\bibliographystyle{IEEEtran}
\bibliography{bibtex/bib/IEEEexample}

\end{document}